\documentclass[a4paper,UKenglish,cleveref, autoref, thm-restate]{lipics-v2021}
\usepackage{tikz-cd}
\usepackage{bbm}
\usepackage{stmaryrd}
\usepackage{comment}
\usepackage{bussproofs}

\DeclareMathOperator{\op}{op}
\DeclareMathOperator{\alg}{Alg}
\DeclareMathOperator{\psh}{PSh}
\DeclareMathOperator{\id}{id}
\DeclareMathOperator{\app}{app}
\def\cat{\EuScript}
\def\ii{\mathrm{i}}
\def\bb{\mathrm{b}}
\def\cc{\mathrm{c}}

\title{What is a Model of the Linear Lambda Calculus?}

\author{Arturo De Faveri}{Université Paris Cité, CNRS, IRIF, F-75013, Paris, France}{defaveri@irif.fr}{https://orcid.org/0009-0004-0115-4500}{}

\authorrunning{Arturo De Faveri} 

\Copyright{Arturo De Faveri} 


\nolinenumbers 

\EventEditors{Corina Cirstea and Yannick Forster}
\EventNoEds{2}
\EventLongTitle{35th EACSL Annual Conference on Computer Science Logic 2027}
\EventShortTitle{CSL 2027}
\EventAcronym{CSL}
\EventYear{2027}
\EventDate{January 26--29, 2027}
\EventLocation{Brighton, UK}
\EventLogo{}
\SeriesVolume{403}
\ArticleNo{10}

\begin{document}

\maketitle

\begin{abstract}
We investigate the notion of model of the linear $\lambda$-calculus from an algebraic perspective. 
Our starting point is the operad of linear $\lambda$-terms, whose algebras provide a natural candidate.
We prove that this notion of model is equivalent to two other structures: a linear analogue of Curry's $\lambda$-algebras,
and semiclosed operads, a class of operads equipped with an internal abstraction operation. 
The equivalence between these three approaches unifies three complementary answers to the question of what should be regarded as a model of the linear $\lambda$-calculus.
As a second contribution, we give a finite equational presentation for the linear variant of $\lambda$-algebras using the linear combinators $\mathbf{B}$, $\mathbf{C}$, and $\mathbf{I}$.  
Finally, exploiting the equivalence with semiclosed operads, we establish a linear analogue of Scott's representation theorem by showing that every model arises as a reflexive object in a monoidal closed category of presheaves.
\end{abstract}

\section{Introduction}
The $\lambda$-calculus \cite{B84,BM22} was introduced by Church in the 1930s as part of his research into the foundations of mathematics.
The problem of its denotational semantics has historically been central and somehow elusive.
The turning point is credited to Dana Scott \cite{Scott}, who proved that models of the $\lambda$-calculus can be elegantly understood in terms of reflexive objects in Cartesian closed categories. 
Scott's approach started from a monoid of composition associated with a specific model, more directly a $\lambda$-algebra, and constructed a Cartesian closed category by taking the category of retracts generated by the monoid. 
This pioneering work subsequently prompted Koymans \cite{K82} and the duo Lambek and Philip Scott \cite[Section I.15]{LS} to analyse in detail what minimal algebraic structure on the monoid of composition was necessary and sufficient to guarantee that the corresponding category would be Cartesian closed. 
However, this Cartesian closed category is obtained by an essentially artificial construction, and it remains an ongoing question whether models of the $\lambda$-calculus can instead be realised as reflexive objects in naturally occurring Cartesian closed categories. 
In this direction, some works have sought to reinterpret classical models of the $\lambda$-calculus from a more abstract categorical perspective \cite{H06}; early contributions to this line of research include \cite{O82,J92,J93}.
In more recent times, Hyland \cite{H17, H14} proposed a paradigm shift, offering a modern ``dress'' to the semantics of the $\lambda$-calculus through the tools of categorical logic and algebraic theories.
In the framework developed by Hyland, a model of the $\lambda$-calculus is just a clone equipped with a \emph{semiclosed} structure. 
The main result of this setup establishes, exploiting a categorical equivalence between semiclosed clones and $\lambda$-algebras, how Scott's representation theorem emerges naturally from the Yoneda lemma.
Recently, categorifications of various aspects of both the typed and untyped $\lambda$-calculus have become an active area of research, leading to a growing body of work \cite{MZ15,Mazza,GuerrieriO21,KerinecMO23,Saville,Amorim}.

Our contribution can be viewed as a counterpart of Hyland's perspective by shifting our focus to the \emph{linear} fragment of the $\lambda$-calculus. 
Unlike the ordinary $\lambda$-calculus, where variables may be duplicated or discarded implicitly, the linear $\lambda$-calculus enforces the principle that each variable is used exactly once. 
Following Girard's introduction of Linear Logic \cite{G87}, the notion of linearity has become a central theme in theoretical computer science \cite{EGRS04,M25}, providing both a semantic and a syntactic discipline for reasoning about resources. 
Since then, linearity has influenced proof theory and categorical semantics \cite{Mellies}, while also inspiring reduction strategies \cite{Guerrini}, type systems \cite{Pfenning} and the design of programming languages \cite{Wadler,Ahmed}.
From both the categorical and the algebraic viewpoints, the linear $\lambda$-calculus calls for replacing Cartesian structures with monoidal ones. 
Symmetric monoidal closed categories provide the appropriate semantics for linear abstraction, while operads \cite{Operads,Leinster} offer the corresponding algebraic language, since their substitution mechanism does not implicitly duplicate or discard variables. 
The collection of linear $\lambda$-terms carries an operad structure, which constitutes the starting point of the present paper. 
Moreover, just as every clone determines a Lawvere theory, every operad gives rise to a symmetric monoidal category, known as its associated PROP. 
The presheaf category over this PROP inherits a symmetric monoidal closed structure via a well-known construction (Day convolution). 
This provides a natural categorical environment in which the algebraic models of the linear $\lambda$-calculus can be represented as reflexive objects, yielding a linear analogue of Hyland's angle on Scott's representation theorem. 

\subparagraph*{Content of the paper}
The paper is organised as follows. 
After reviewing some preliminaries (Section 2), we study the category of algebras of the operad $L$ of linear $\lambda$-terms, which provides our first notion of model for the linear $\lambda$-calculus (Section 3). 
In Section 4, we show that these algebras admit an equational description as a distinguished subvariety of $\mathbf{BCI}$-algebras--a variety generated by the linear combinators $\mathbf B$, $\mathbf C$, and $\mathbf I$--thereby introducing a linear analogue of Curry's notion of $\lambda$-algebra. 
In particular, we identify a finite set of equations characterising precisely those $\mathbf{BCI}$-algebras that are sound and complete with respect to the $\beta$-theory of the linear $\lambda$-calculus (Theorem \ref{thm:equiv}). 
In Section 5, we study semiclosed operads and prove that the operad $L$ is initial among them. 
Semiclosed operads provide the third notion of model for the linear $\lambda$-calculus.
The main technical result of the paper is the equivalence between the three resulting categories: $L$-algebras, linear $\lambda$-algebras, and semiclosed operads.
This is developed in Section 6 and culminates in Theorem \ref{thm:final}. 
We exploit this correspondence to establish a linear analogue of Scott's representation theorem, proving that every model of the linear $\lambda$-calculus can be represented as a reflexive object in a monoidal closed category of presheaves (Theorem \ref{thm:scott}). 

\subparagraph*{Related work}
The linear $\lambda$-calculus has attracted sustained attention over the past decades. 
Already in the early 1990s, Jacobs \cite{J93} recognised symmetric monoidal closed categories as the natural categorical setting for substructural $\lambda$-calculi and introduced the corresponding notion of reflexive object. 
Categorical methods have been employed to provide a uniform treatment of abstract syntax with binding and substitution across the spectrum of substructural calculi \cite{FR25}.
Recently, Hasegawa has investigated a \emph{braided} variant of the linear $\lambda$-calculus \cite{H20,H23}.
To draw a comparison with the present paper, \cite{H23} introduces an equivalent formulation of the notion of semiclosed operad and establishes a finite equational axiomatisation of the variety of $\mathbf{BCI}$-algebras that is sound and complete for the $\beta\eta$-theory of the linear $\lambda$-calculus.
It is also possible to define an analogue of $\mathbf{BCI}$-algebra for the \emph{planar} $\lambda$-calculus, called $\mathbf{BI}(-)^{\bullet}$-algebra and developed by Tomita \cite{Tomita,Tomita2}.  
The study of the linear $\lambda$-calculus has also uncovered unexpected links with seemingly distant areas. 
A notable example is the work of Zeilberger \cite{ZG15,Z16}, who established surprising connections between linear $\lambda$-terms and geometric combinatorics. 
A major source of inspiration for the present paper is Hyland's algebraic-theoretic approach to the semantics of the $\lambda$-calculus \cite{H17}. 
However, we depart from Hyland by establishing a direct equivalence between $L$-algebras and semiclosed operads, avoiding the introduction of an intermediate notion analogous to Hyland's $\lambda$-monoid \cite{H14}. 
Furthermore, our definition of linear $\lambda$-algebra is indebted to Selinger's conceptual analysis of Curry's notion of $\lambda$-algebra given in his surveys \cite{S02,S13}.
Also closely related to our work is the recent paper \cite{KNRS}, 
which characterises combinatory completeness, relative to a choice of a faithful Cartesian club,
as closure of the computable maps under composition in the multicategory. 
An applicative system is a $\mathbf{BCI}$-algebra if and only if all linear polynomials are computable,
and its computable maps then form a \emph{weakly-closed} operad, a notion strictly weaker than the semiclosed operads presented here. 
Linear $\lambda$-algebras refine $\mathbf{BCI}$-algebras, refinement that corresponds to the
passage from weakly-closed to semiclosed operads; we thus establish, in the linear case,
the kind of correspondence between ``$\lambda$-models'' and semiclosed operads that~\cite[Section 7]{KNRS} conjectures.

\section{Preliminaries}
We recall some basic notions and results concerning monoidal categories \cite[Chapters VII, XI]{ML}, operads \cite{Operads,Leinster}, and the $\lambda$-calculus \cite{B84,BM22}.
We assume that the reader is familiar with end and coend calculus \cite{L21}. 
Throughout the paper, we denote by $\mathbb{N}$ the set of nonnegative integers. 

\subsection{Monoidal categories}
\label{sec:monoidal}
A \emph{monoidal} category is a category $\EuScript{C}$ together with a bifunctor $-\otimes- : \EuScript{C}^2 \to \EuScript{C}$, called \emph{tensor}, an object $I \in \EuScript{C}$, and natural isomorphisms: 
\begin{equation*}
    \alpha_{a,b,c}: (a \otimes b) \otimes c  \to a \otimes (b \otimes c), \quad \lambda_a: I \otimes a \to a, \quad \rho_a : a \otimes I \to a,
\end{equation*}
for all $a,b,c \in \EuScript{C}$, 
subject to some coherence conditions that can be found for instance in \cite[p. 151]{R20}. 
If $\alpha, \lambda, \rho$ are all identities, the monoidal category $\EuScript{C}$ is said to be \emph{strict}. 

Let $\EuScript{C}$ and $\EuScript{D}$ be two monoidal categories. 
A functor $F: \EuScript{C} \to \EuScript{D}$ is a \emph{(lax) monoidal functor} if for each pair $a, b \in \EuScript{C}$ there is a morphism $\phi_{a,b} : F(a) \otimes F(b) \to F(a \otimes b)$, natural in $a,b$ and a morphism $\eta: I \to F(I)$, 
well-behaved with respect to $\alpha, \lambda$ and $\rho$ \cite[p. 161-162]{R20}. 
We say that $F$ is \emph{strong monoidal} if $\phi_{a,b}$ and $\eta$ are natural isomorphisms.  
We say that $F$ is \emph{strictly monoidal} if $\phi_{a,b}$ and $\eta$ are identities.  

A monoidal category $(\cat C, \otimes, I)$ is \emph{symmetric} if for all $a,b \in \cat C$ there is a natural isomorphism $\tau_{a,b} : a \otimes b \simeq b \otimes a$ 
satisfying some compatibility conditions \cite[p. 158]{R20}.

\begin{definition}
    A (symmetric) monoidal category $\EuScript{C}$ is \emph{closed} if for every $c \in \EuScript{C}$ the functor $- \otimes c : \EuScript{C} \to \EuScript{C}$ admits a right adjoint, i.e. a functor $c \multimap -: \EuScript{C} \to \EuScript{C}$ witnessing a natural isomorphism
        $\EuScript{C}(b \otimes c, d) \simeq \EuScript{C}(b, c \multimap d)\text{,}$
    for every $b,d \in \cat C$. 
\end{definition}

An important example of symmetric monoidal products is given by Cartesian products. 
In this case the monoidal category $\cat C$ is called \emph{Cartesian}, and \emph{Cartesian closed} if moreover $\cat{C}$ is closed.   

\begin{definition} [{\cite[Definition 5.3]{J93}}]
    \label{def:reflexive}
    Let $\cat C$ be a closed symmetric monoidal category. 
    An object $x \in \cat C$ is \emph{linear reflexive} if $x \multimap x$ is a retract of $x$, i.e., there are morphisms
    $r:x \to (x \multimap x )$ and $s:(x \multimap x) \to x$ such that $r \circ s = \id_{x \multimap x}$. 
\end{definition}

Let $ (\cat C, \otimes, I)$ be a symmetric monoidal category. 
The category of presheaves $\psh(\cat C)$ on $\cat C$ canonically inherits a monoidal structure via \emph{Day's convolution product} \cite{Day}:
given $F, G : \cat C^{\op} \to \mathbf{Set}$ we define 
\begin{equation}
    \label{eq:day}
    (F \otimes G)(c):=\int^{x,y \in \cat C} \cat C(c,x \otimes y) \times F(x) \times G(y)
\end{equation}
for every $c \in \cat C$. 
The unit of this product is given by the representable presheaf $J(c):=\cat C(c,I)$ for every $c \in \cat C$. 
This product makes the Yoneda embedding $\cat C \to \psh (\cat C)$ strong monoidal, i.e., $\cat C(-,a) \otimes \cat C(-,b) \simeq \cat C(-, a \otimes b)$. 
Finally, this structure is actually monoidal \emph{closed} with exponentials given by 
\begin{equation}
    (F \multimap G)(c) = \int_{x \in \cat C} \mathbf{Set}(F(x),G(c \otimes x))
\end{equation}
for every $c \in \cat C$.

\subsection{Operads and PROPs}
Operads, defined for the first time by May \cite{M72}, originated in algebraic topology as a tool to describe algebraic structures that are invariant under homotopy \cite{BV73}.
We stress that here we are interested in \emph{symmetric}, and not in \emph{planar}, operads.

Let $S_n$ be the symmetric group on $n$ elements and let $k_1, \ldots, k_n \in \mathbb{N}$.  
Given $\sigma_i \in S_{k_i}$ for $1 \le i \le n$, let $[\sigma_1, \ldots, \sigma_n] \in S_{k_1 + \cdots + k_n}$ be the permutation that permutes the first $k_1$ elements by $\sigma_1$, 
the next $k_2$ elements by $\sigma_2$, etc., and keeps the order of the blocks intact.
Given $\sigma \in S_n$, let $\hat{\sigma} \in S_{k_1 + \cdots +k_n}$ be the permutation that acts by breaking $\{1, \ldots, k_1 + \cdots +k_n\}$ into 
$n$ blocks, the first of size
$k_1$, the second of size 
$k_2$, through the 
$n$-th block of size 
$k_n$, and then permutes these 
blocks by $\sigma$, keeping each block intact. 

\begin{definition}
    \label{def:operad}
    An \emph{operad} $O$ is a sequence of sets $\{O(n) : n \in \mathbb N\}$ together with an element $1 \in O(1)$ called the \emph{identity} and \emph{composition} maps
    \begin{equation*}
        \circ: O(n) \times O(k_1) \times \cdots \times O(k_n) \to O(k_1 + \cdots + k_n), \, (f, g_1, \ldots, g_n) \mapsto f \circ (g_1, \ldots, g_n), 
    \end{equation*}
    satisfying, writing $f(g_1, \ldots, g_n)$ in place of $f \circ (g_1, \ldots, g_n)$, the following identities: 
    \begin{enumerate}
        \item $f(1, \ldots, 1) = f$;
        \item $1(f)=f$; 
        \item $f(g_1(h^1_1, \ldots, h^1_{k_1}), \ldots, g_n(h^n_1, \ldots, h^n_{k_n})) =f(g_1, \ldots, g_n) \circ (h^1_1, \ldots, h^1_{k_1}, \ldots, h^n_1, \ldots, h^n_{k_n})$. 
    \end{enumerate}
    Moreover, each $O(n)$ is endowed with a right action of $S_n$, $(\sigma, f) \mapsto f\sigma$, satisfying: 
    \begin{enumerate}
        \setcounter{enumi}{3}
        \item $(f\sigma )(g_1, \ldots, g_n) = f(g_{\sigma^{-1}1}, \ldots, g_{\sigma^{-1}n})\hat{\sigma}$; 
        \item $f(g_1\sigma_1, \ldots,  g_n\sigma_n) = f(g_1, \ldots, g_n)[\sigma_1, \ldots, \sigma_n]$. 
    \end{enumerate}
    A morphism of operads $\phi: O \to O'$ is given by a sequence of maps $\phi_n : O(n) \to O'(n)$, $n \in \mathbb{N}$, such that $\phi_1(1)=1$, $\phi_n(f\sigma)=\phi_n(f)\sigma$ and 
    $\phi_{k_1 + \cdots + k_n}(f(g_1, \ldots, g_n)) = \phi_n(f) (\phi_{k_1}(g_1), \ldots, \phi_{k_n}(g_n))\text{.}$
\end{definition}

\begin{example}
    \label{ex:endo}
    Let $\cat C$ be a symmetric monoidal category, and let $x \in \cat C$. 
    Then the sequence of sets $O(n):=\cat C(x^{\otimes n}, x)$, $n \in \mathbb{N}$, is an operad called the \emph{endomorphism operad} of $x$. 
\end{example}

\begin{definition}
    \label{def:algebra}
    Let $O$ be an operad. 
    An \emph{algebra} or a \emph{model} for $O$ is a set $A$ together with maps (called actions) $\cdot: O(n) \times A^n \to A$ for all $n \in \mathbb{N}$ that satisfy the following conditions: 
    \begin{enumerate}
        \item $f(g_1, \ldots, g_n) \cdot (a^1_1, \ldots ,a^1_{k_1}, \ldots, a^n_1, \ldots, a^n_{k_n}) = f \cdot (g_1 \cdot (a^1_1, \ldots, a^1_{k_1}), \ldots ,g_n \cdot (a^n_1, \ldots, a^n_{k_n}))$; 
        \item $1 \cdot a = a$; 
        \item $f \cdot (a_1, \ldots, a_n) = (f \sigma) \cdot (a_{\sigma^{-1} 1}, \ldots, a_{\sigma^{-1} n})$ for every $\sigma \in S_n$. 
    \end{enumerate}
    A morphism of $O$-algebras $A$ and $B$ is a function $\alpha : A \to B$ such that 
    $\alpha(f \cdot (a_1, \ldots, a_n))= f \cdot (\alpha(a_1), \ldots, \alpha(a_n))$. 
    Given an operad $O$, we write $\alg(O)$ for the category of $O$-algebras. 
\end{definition}

\begin{definition}
  \label{def:module}
 Let $O$ be an operad. 
 A \emph{module} over $O$ is a sequence of sets $\{X(n): n \in \mathbb{N}\}$ together with maps 
 \begin{equation*}
  X(n)\times O(k_1)\times\cdots\times O(k_n)\to X(k_1+\cdots+k_n), \, (x,f_1,\dots,f_n)\mapsto x(f_1,\dots,f_n)
 \end{equation*}
 satisfying: 
    \begin{enumerate}
        \item $x(1, \ldots, 1) = x$;
        \item $x(g_1(h^1_1, \ldots, h^1_{k_1}), \ldots, g_n(h^n_1, \ldots, h^n_{k_n})) =  \\$
        $(x(g_1, \ldots, g_n) )(h^1_1, \ldots, h^1_{k_1}, \ldots, h^n_1, \ldots, h^n_{k_n})$. 
    \end{enumerate}
 Moreover, there is a right action of $S_n$, for every $n$, $(x,\sigma)\mapsto x\sigma$ satisfying: 
  \begin{enumerate}
        \setcounter{enumi}{2}
        \item $(x\sigma )(g_1, \ldots, g_n) = x(g_{\sigma^{-1}1}, \ldots, g_{\sigma^{-1}n})\hat{\sigma}$; 
        \item $x(g_1\sigma_1, \ldots,  g_n\sigma_n) = x(g_1, \ldots, g_n)[\sigma_1, \ldots, \sigma_n]$. 
    \end{enumerate}
 A morphism of $O$-modules $\phi : X\to Y$ is a family of maps
 $\phi_n : X(n)\to Y(n)$ such that $\phi(p\sigma)=\phi(p)\sigma$ and
 $\phi (x(f_1,\dots,f_n))=\phi(x)(f_1,\dots,f_n)$. Modules and their
 morphisms form a category. 
\end{definition}

Clearly, $O$ is a module over itself.

\begin{definition}
    A \emph{PROP} is a strict symmetric monoidal category $\cat P$ whose objects are the nonnegative integers, with unit $I:=0$ and monoidal product $n \otimes m := n+m$ for every $n, m \in \mathbb{N}$. 
\end{definition}

Given a group $G$ that acts on the right on a set $X$ and on the left on a set $Y$ we denote by $X \times_G Y$ the quotient $X \times Y / \sim$ such that 
$(xg, y) \sim (x,gy)$ for every $x \in X, y \in Y$ and $g \in G$. 
Every operad $O$ gives rise to a PROP $\cat O$ in the following way. 
This construction is well-known and further details can be found in \cite[Section 2.3]{Leinster}. 
    For every $n,m \in \mathbb N$, we define 
    \begin{equation}
        \cat O(n,m):= \coprod_{k_1 + \cdots + k_m = n} (O(k_1) \times \cdots \times O(k_m)) \times_{S_{k_1}\times \cdots \times S_{k_m}}  S_n 
    \end{equation}
where we observe that $S_{k_1}\times \cdots \times S_{k_m}$ acts on $S_n$ on the left by
multiplication by the image of the natural inclusion of $S_{k_1}\times \cdots \times S_{k_m}$
into $S_n$.

Let $O$ be an operad and let $\cat O$ be the corresponding PROP. 
The category of algebras for $O$ can equivalently be described as the category of strict monoidal functors $\cat O \to \mathbf{Set}$. 

From the fact that the Yoneda functor is strong monoidal and from the Yoneda lemma we have the following. 
\begin{lemma}
    \label{lem:endo}
    Let $O$ be an operad and $\cat O$ be the associated \emph{PROP}.  
    Then $O$ is isomorphic to the endomorphism operad of the presheaf $\cat O(-,1)$ in $\psh(\cat O)$. 
\end{lemma}

\subsection{The lambda calculus}
The set $\Lambda$ of \emph{$\lambda$-terms} over a countably infinite set $X$ of variables is defined by the grammar: 
$M,N ::= x\mid \lambda x.M\mid MN$, for $x \in X$. 
We assume that application associates to the left and has higher precedence than abstraction, and we abbreviate successive $\lambda$-abstractions in the usual way:
for instance, $\lambda xyz.xyz$ stands for $\lambda x.(\lambda y.(\lambda z.(xy)z))$.
From now on, $\lambda$-terms are considered up to renaming of bound variables. 
A $\lambda$-term without free variables is called a \emph{combinator}.
The set of combinators is denoted by $\Lambda(0)$. 
The \emph{$\beta$-reduction} $\to_\beta$ and \emph{$\eta$-reduction} $\to_\eta$ are defined as the closures of the rules: 
\begin{equation*}
    (\lambda x.M)N \to_\beta M[N/x], \qquad \qquad \lambda x.Mx \to_\eta M, \,   \text{if $x$ does not occur free in $M$}
\end{equation*}
where we denote by $M[N/x$] the \emph{capture-free substitution} of $N$ for all free occurrences of $x$ in $M$. 
We define \emph{$\beta\eta$-reduction} as $\to_{\beta \eta}\ =\ \to_\beta \cup \to_\eta$. 
We denote by $=_\beta$ and $=_{\beta \eta}$ the smallest equivalence relation on $\Lambda$ containing $\to_\beta$ and $\to_{\beta \eta}$, respectively. 

\begin{example}
    \label{ex:lincomb}
       The combinators below are used throughout the paper: 
    \begin{equation*}
        \mathbf{I}= \lambda x.x, \qquad \mathbf{1}=\lambda xy.xy, \qquad \mathbf{B} = \lambda fgx. f(gx),\qquad \mathbf{C} = \lambda fxy.fyx \text{.}
    \end{equation*} 
    We recall that $\mathbf{I}$ is the identity, $\mathbf{1}$ the 1st Church numeral, $\mathbf{B}$ stands for the composition and $\mathbf{C}$ for the swapping. 
\end{example}

\section{The linear lambda calculus}
\begin{definition}
    A $\lambda$-term $M$ is called \emph{linear} if every free variable has exactly one occurrence in $M$, and every bound variable has exactly one occurrence in the scope of its abstraction.
    We denote by $\Lambda_{lin}$ the set of linear $\lambda$-terms. 
\end{definition}

Observe that $\to_{\beta}$ induces a reduction on $\Lambda_{lin}$ since if $M \in \Lambda_{lin}$ and $M \to_{\beta} N$, then $N \in \Lambda_{lin}$.
The combinators defined in Example \ref{ex:lincomb} are all linear. 
If $X \subseteq \Lambda(0)$ is a set of combinators, we denote by $X^+$ the least set $Y \subseteq \Lambda(0)$ such that 
$X \subseteq Y$ and if $M,N \in Y$ then $MN \in Y$. 
A set $B \subseteq \Lambda(0)$ is a \emph{basis} for a set of combinators $X$ (supposed to be closed under $=_{\beta}$) if for every $M \in X$ there is $N \in B^+$ such that $M =_{\beta} N$. 
For example, $\{\mathbf{K}:=\lambda xy.x, \mathbf{S}:=\lambda xyz.xy(xz)\}$ is a basis for $\Lambda(0)$ \cite[Proposition 8.1.2]{B84}. 
We spell out the details of the following well-known fact \cite[Section 9F]{Hindley}. 

For $x_1,\dots,x_n$ distinct variables, let $\mathbf{BCI}\langle x_1, \ldots, x_n\rangle$ be the set of $\lambda$-terms
generated by the grammar $P,Q::=x_i\mid \mathbf{B}\mid \mathbf{C}\mid \mathbf{I}\mid PQ$ in which each $x_i$ occurs exactly once.

\begin{lemma}
    \label{lem:hindley}
Let $M\in\Lambda_{lin}$ with set of free variables exactly $\{x_1,\dots,x_n\}$. 
Then there is $N\in\mathbf{BCI}\langle x_1, \ldots, x_n\rangle$ such that $M=_\beta N$. 
In particular, the set $\{\mathbf{B}, \mathbf{C}, \mathbf{I}\}$ is a basis for the linear combinators $\Lambda_{lin}(0)$. 
\end{lemma}

\begin{proof}
The statement follows from the following claim:
if $N\in \mathbf{BCI}\langle x_1, \ldots, x_n,x\rangle$ then there is $N'\in\mathbf{BCI}\langle x_1, \ldots, x_n\rangle$ with $\lambda x.N=_\beta N'$.
We argue by induction on $N$. 
Since $x$ occurs in $N$, either $N=x$, and $\lambda x.x=\mathbf{I}$, or $N=PQ$ with $x$ occurring in exactly one of $P,Q$. 
If $x$ occurs in $P$, then
$\mathbf{C}(\lambda x.P)Q\to_\beta\lambda y.(\lambda x.P)yQ\to_\beta\lambda y.P[y/x]Q = \lambda x.PQ$.
By induction $\lambda x.P=_\beta P'$ with $P'\in\mathbf{BCI}\langle x_1, \ldots, x_n\rangle$, so $\lambda x.PQ=_\beta \mathbf{C}P'Q\in\mathbf{BCI}\langle x_1, \ldots, x_n\rangle$.
If $x$ occurs in $Q$, then 
$\mathbf{B}P(\lambda x.Q)\to_\beta\lambda z.P((\lambda x.Q)z)\to_\beta\lambda z.P(Q[z/x]) = \lambda x.PQ$,
and we conclude as before with $\mathbf{B}PQ'$ for some $Q' \in \mathbf{BCI}\langle x_1, \ldots, x_n\rangle$, $ Q' =_\beta\lambda x.Q$. 
\end{proof}

Inspired by the previous result one can give the following definition.

\begin{definition} [{\cite[Definition 2]{Hoshino}}]
    \label{def:bci}
    A set $A$ together with a binary operation $\cdot : A^2 \to A$ and three nullary operations $\bb,\cc,\ii \in A$ satisfying
    \begin{equation*}
        \bb xyz = x(yz), \quad  \cc xyz = xzy, \quad \ii x = x
    \end{equation*}
    is called a \emph{$\mathbf B \mathbf C \mathbf I$-algebra}. 
    As usual, we followed the convention of writing $xy$ for $x \cdot y$ and of associating to the left.
    Moreover, we let $1:=\bb \ii$. 
    A homomorphism of $\mathbf{BCI}$-algebras $A$ and $A'$ is a function $f: A \to A'$ such that $f(\bb)=\bb$, $f(\cc)=\cc$, $f(\ii)=\ii$ and such that $f(xy)=f(x)f(y)$ for all $x,y \in A$.
    We will denote by $\mathbf{BCI}$ the category of $\mathbf{BCI}$-algebras and $\mathbf{BCI}$-algebra homomorphisms. 
\end{definition}

\subsection{The operad of linear lambda terms}
A \emph{context} is a finite sequence of pairwise distinct variables. 
We write $\Gamma, \Delta$ for the concatenation of two contexts $\Gamma$ and $\Delta$, with the implicit assumption that their supports are disjoint. 
Let $\Gamma \vdash M$ be the relation between contexts and $\lambda$-terms defined inductively by the following type system:
\vspace{-0.4\baselineskip}
  \begin{center}
    \AxiomC{}
    \RightLabel{$\qquad \qquad$}
    \UnaryInfC{$x \vdash x$}
    \DisplayProof
    \AxiomC{$\Gamma \vdash M$}
    \AxiomC{$\Delta \vdash N$}
    \RightLabel{$\qquad \qquad$}
    \BinaryInfC{$\Gamma, \Delta \vdash MN$}
    \DisplayProof
    \AxiomC{$\Gamma, x \vdash M$}
    \RightLabel{$\qquad \qquad$}
    \UnaryInfC{$\Gamma \vdash \lambda x. M$}
    \DisplayProof
    \AxiomC{$\Gamma, x, y , \Delta \vdash M$}
    \UnaryInfC{$\Gamma, y, x, \Delta \vdash M$}
    \DisplayProof
  \end{center}
  \vspace{-0.3\baselineskip}
Linear $\lambda$-terms are precisely those $\lambda$-terms that can be typed in this type system. 
For each $n \in \mathbb{N}$, let $L(n)$ be the set of judgements $\Gamma \vdash M$ such that the support of $\Gamma$ has cardinality $n$,
up to renaming of free variables and up to $\beta$-equivalence. 

\begin{lemma}
   The sequence of sets $\{L(n) : n \in \mathbb{N}\}$ is an operad $L$ with:
    \begin{itemize}
        \item identity $1:=x_1 \vdash x_1 \in L(1)$;
        \item if $x_1, \ldots, x_n \vdash M \in L(n)$ and $\Delta_i \vdash N_i \in L(k_i)$ for $1 \le i \le n$, 
        $M(N_1, \ldots, N_n) := \Delta_1 , \ldots, \Delta_n \vdash M[N_1/x_1, \ldots, N_n/x_n]\text{;}$
        \item if $\sigma \in S_n$ and $x_1, \ldots, x_n \vdash M$, $M\sigma:= x_{\sigma 1}, \ldots, x_{\sigma n} \vdash M$. 
    \end{itemize}
\end{lemma}

The algebras $\alg(L)$ for the operad $L$ can be seen as $\mathbf{BCI}$-algebras. 
Indeed, given an $L$-algebra $A$, we endow $A$ with the structure of $\mathbf{BCI}$-algebra as follows. 
Let $\app:=x_1, x_2 \vdash x_1x_2 \in L(2)\text{.}$
We define the application $a_1 a_2:=\app \cdot (a_1,a_2)$ as the action of the linear term $x_1x_2$ on the pair $(a_1, a_2) \in A^2$. 
The nullary operations $\bb,\cc, \ii \in A$ are given by the action of the combinators $\mathbf{B}, \mathbf{C}, \mathbf{I}$ on the empty sequence $() \in A^0$: 
\begin{equation*}
    \bb:= \mathbf{B} \cdot (), \quad \cc:= \mathbf{C} \cdot (), \quad \ii:= \mathbf{I} \cdot ()\text{.}
\end{equation*}

Then one has to check that the identities of Definition \ref{def:bci} are satisfied. 
For the sake of an example, observe that, by definition of the action, for all $a_1, a_2, a_3 \in A$ 
\begin{align*}
    \app \cdot (\app \cdot (\app \cdot (\mathbf{B} \cdot (), a_1), a_2), a_3) & = 
    \app \cdot (\app \cdot (\lambda x_2 x_3. x_1 (x_2 x_3) \cdot a_1, a_2 ), a_3) \\
    & = \app \cdot ( \lambda x_3. x_1 (x_2 x_3) \cdot (a_1, a_2), a_3) \\
    & = x_1 (x_2 x_3) \cdot (a_1, a_2, a_3)\\
    & = \app \cdot (a_1, \app \cdot (a_2, a_3))\text{.}
\end{align*}

\begin{proposition}
    \label{prop:ff}
    There is a fully faithful functor $G:\alg(L) \to \mathbf{BCI}$. 
\end{proposition}
\begin{proof}
    If $A$ is an $L$-algebra we define $G(A)$ to be the $\mathbf{BCI}$-algebra with underlying set $A$ defined above. 
Let $f: A \to B$ be a morphism of $L$-algebras. 
We show that $Gf:=f$ is a homomorphism of $\mathbf{BCI}$-algebras. 
Since the following diagrams are commutative, 
 \begin{center}
    \hspace*{-4cm}
        \begin{tikzcd}
            L(0)  \arrow[r, "\cdot"] \arrow[d, equal] & A \arrow[d, "f"] & & L(2) \times A^2  \arrow[r, "\cdot"] \arrow[d, "\id \times f^2"] & A \arrow[d, "f"] \\
            L(0)  \arrow[r, "\cdot"]  & B & &  L(2) \times B^2  \arrow[r, "\cdot"]  & B \\
        \end{tikzcd}
    \end{center}
    \vspace{-0.6\baselineskip}
$f$ preserves the nullary operations and $f(\app \cdot(a_1, a_2)) = \app \cdot (f(a_1), f(a_2))$ for all $a_1, a_2 \in A$. 
We now show that $G$ is fully faithful. 
The faithfulness is obvious since $G$ acts as the identity on the arrows. 
The functor $G$ is full, that is: if $f:A \to B$ preserves the binary operation of application and the nullary operations $\mathbf{B}, \mathbf{C}, \mathbf{I}$, then $f$ preserves the whole action of $L$. 
Let $f$ be such a map. 
Let $M \in \mathbf{BCI} \langle x_1, \ldots, x_n \rangle$ and write $\bar{M}$ for the equivalence class $x_1, \ldots, x_n \vdash M \in L(n)$. 
We show that 
\begin{equation}
    \label{eq:claim}
    f(\bar{M} \cdot (a_1, \ldots, a_n) )= \bar{M} \cdot (f(a_1), \ldots , f(a_n))
\end{equation}
for every $a_1, \ldots a_n \in A$ by induction on $M$. 
If $n=1$ and $M=x_1$, then $f(x_1 \cdot a) = f(a) = x_1 \cdot f(a)$ for all $a \in A$. 
If $M=\mathbf{B}, \mathbf{C}$ or $\mathbf{I}$, then \eqref{eq:claim} holds by assumption. 
Let $M=PQ$; then $\bar{M}= \app(\bar{P}, \bar{Q})$ for some $\bar{P} \in L(i)$, $\bar{Q} \in L(j)$ with $i +j=n$ (for simplicity we assume that there is no permutation of the variables) and, again, \eqref{eq:claim} holds by assumption and by induction.  
Finally, Lemma \ref{lem:hindley} shows that every element of $L(n)$ is the class $\bar{M}$ of some $M \in \mathbf{BCI} \langle x_1, \ldots, x_n \rangle$.  
\end{proof}

\section{Linear lambda algebras}
In this section we study the essential image of the functor of Proposition \ref{prop:ff}. 
We characterise this category as a subvariety of the variety of $\mathbf{BCI}$-algebras. 

\subsection{An equational characterisation}
Given a $\mathbf{BCI}$-algebra $A$, let $\Sigma_A$ be the signature $\Sigma \cup A$ where $\Sigma$ is the signature of $\mathbf{BCI}$-algebras, and each $a \in A$ is nullary. 
Let $T^{lin}_{\Sigma_A}(x_1, \ldots, x_n)$ be the set of linear terms generated by the grammar 
$p,q::=_{lin}x_i \mid a \mid  pq $
with $a \in A$ and $1 \le i \le n$.
If a variable $x_i$ occurs in $p \in T^{lin}_{\Sigma_A}(x_1, \ldots, x_n)$, then $x_i$ occurs exactly once in $p$. 

Given a $\mathbf{BCI}$-algebra $A$, let $A[x_1, \ldots, x_n]$ be the free extension of $A$ by the indeterminates $x_1, \ldots, x_n$. 
This is the $\mathbf{BCI}$-algebra that satisfies the following universal property. 
    There is an inclusion $A \hookrightarrow A[x]$ and for every $\mathbf{BCI}$-algebra homomorphism $f: A \to A'$ and every $a' \in A'$, there is a unique homomorphism $\hat{f} : A[x] \to A'$ such that 
    $\hat{f}(a)=f(a)$ for all $a \in A$ and $\hat{f}(x)=a'$. 
Let $A_{lin}[x_1, \ldots, x_n]$ be the subalgebra of $A[x_1,\ldots, x_n]$ generated by the linear terms in $T^{lin}_{\Sigma_A}(x_1, \ldots, x_n)$. 
The algebra $A_{lin}[x_1, \ldots, x_n]$ can be obtained by quotienting the set $T^{lin}_{\Sigma_A}(x_1, \ldots, x_n)$ with the linear equations of $\mathbf{BCI}$-algebra and the linear equations holding in $A$. 
By abuse, we shall not distinguish notationally between a term in $T^{lin}_{\Sigma_A}(x_1, \ldots, x_n)$ and its equivalence class in $A_{lin}[x_1, \ldots, x_n]$.

Now, for each $t \in T^{lin}_{\Sigma_A}(x_1, \ldots, x_n,x)$, we define $\lambda^*x.t \in T^{lin}_{\Sigma_A}(x_1, \ldots, x_n)$ as follows by induction on $t$:
\begin{align*}
    \lambda^*x.x &:= \ii, & \\
    \lambda^*x.pq &:= \cc(\lambda^* x.p)q & \text{if $x$ is free in $p$,}\\
    \lambda^* x .pq & := \bb p (\lambda^* x. q) & \text{if $x$ is free in $q$.}
\end{align*}
Thus the significance of the combinators $\mathbf{B}, \mathbf{C}$, and $\mathbf{I}$ lies in the fact that they can simulate linear $\lambda$-abstraction. 

This definition of $\lambda^*$ satisfies the $\beta$-equivalence rule, i.e., $A$ satisfies $(\lambda^*x.t)a = t[a/x]$, for all $a \in A$ and $t \in T^{lin}_{\Sigma_A}(x)$.
Indeed, by induction on $t$: 
\begin{align*}
    (\lambda^* x.x)a & = \ii a = a, &\\
    (\lambda^*x.pq)a & = \cc(\lambda^*x.p)qa = (\lambda^* x.p) a q = p[a/x] q &\text{if $x$ is free in $p$,}\\
    (\lambda^*x.pq)a & = \bb p(\lambda^*x.q)a = p((\lambda^*x.q)a)= p (q[a/x])  &\text{if $x$ is free in $q$.}
\end{align*}
Note that the last cases use the induction hypothesis for $p$ and $q$.

Given a $\mathbf{BCI}$-algebra $A$, we can interpret any linear $\lambda$-term $M$ with free variables $x_1, \ldots, x_n$ in $A_{lin}[x_1, \ldots, x_n]$ by inductively defining 
$\llbracket M \rrbracket \in T^{lin}_{\Sigma_A}(x_1, \ldots, x_n)$:
\begin{align*}
    \llbracket x_i \rrbracket & := x_i, \\
    \llbracket M N \rrbracket & := \llbracket M \rrbracket \llbracket N \rrbracket, \\
    \llbracket \lambda x.M \rrbracket & := \lambda^* x. \llbracket M \rrbracket. 
\end{align*}
We shall omit the brackets $\llbracket - \rrbracket$ if it does not create any confusion. 

\begin{definition}
    A $\mathbf{BCI}$-algebra $A$ is a \emph{linear $\lambda$-algebra} if $A$ satisfies 
    \begin{equation}
        \label{eq:bc-derived}
        \lambda^* xyz. x(yz) = \bb, \quad \lambda^* xyz. xzy = \cc
    \end{equation}
    and if for every $M,N \in \Lambda_{lin}(n)$, 
    \begin{equation}
        \label{eq:def-linear}
        M =_{\beta} N \quad \text{implies that} \quad \llbracket M \rrbracket = \llbracket N \rrbracket
    \end{equation} 
    in $A_{lin}[x_1, \ldots, x_n]$.
\end{definition}

\begin{remark}
    \label{rem:comm}
    Observe that the operation $\llbracket - \rrbracket$ commutes with substitution 
(that is, for every $M,N \in \Lambda_{lin}$, $\llbracket M[N/x] \rrbracket = \llbracket M \rrbracket [\llbracket N \rrbracket / x]$). 
\end{remark}

\begin{table}[t]
  \caption{Equations whose $\lambda^*$-closure defines linear $\lambda$-algebras; recall that $1 = \bb \ii$.}
    \label{tab:linear}
\centering
\begin{tabular}{c c c}

$1\bb = \bb$  & $1(\bb x)=\bb x$ & $1(\bb xy)=\bb xy$  \\
$1\cc = \cc$  & $1(\cc x)=\cc x$ & $1(\cc xy)=\cc xy$ \\
$1\ii = \ii$  &  $\bb \ii x = \bb x \ii$ &  \\
$\cc(\cc(\bb \bb x)y)z = \cc x(yz)$ & $\cc (\bb (\bb x)y)z = \bb  x(\cc yz)$ & $\bb (\bb xy)z=\bb x(\bb yz)$ \\
$\cc (\cc (\bb \cc x)y)z = \cc(\cc xz)y$ & $\cc(\bb(\cc x)y)z = \bb(xz)y$ & $\bb(\cc xy)z=\cc(\bb xz)y$ 
\end{tabular}
\end{table}

Given an identity of Table \ref{tab:linear}, its $\lambda^*$-\emph{closure} is the identity obtained by $\lambda^*$-abstracting every free variable in it.

\begin{lemma}
    \label{lem:table}
    Any linear $\lambda$-algebra satisfies the identities of Table \ref{tab:linear}. 
    Moreover, any linear $\lambda$-algebra satisfies their $\lambda^*$-closure.  
\end{lemma}
\begin{proof}
    By the definition of $\llbracket - \rrbracket$, and by \eqref{eq:def-linear}, the proof of the first statement is a
    lengthy but straightforward calculation. 
    For example, as 
        \begin{align*}
        \mathbf{B}(\mathbf{C}xy)z  =_{\beta} \lambda u. \mathbf{C}xy(zu) 
         =_{\beta} \lambda u. x(zu) y 
         =_{\beta} \lambda u. (\mathbf{B}xz) uy 
         =_{\beta} \mathbf{C}(\mathbf{B}xz) y\text{,}
    \end{align*}
    any linear $\lambda$-algebra satisfies $\bb(\cc xy)z=\cc(\bb xz)y$. 
    The other identities are proved similarly.
    The second statement follows from a lengthier, but still straightforward calculation. 
\end{proof}

\begin{lemma}
    \label{lem:closure}
    If $A$ is a $\mathbf{BCI}$-algebra that satisfies the $\lambda^*$-closure of an identity, then $A$ satisfies also the identity itself. 
\end{lemma}
\begin{proof}
    Assume that $A$ satisfies $\lambda ^*x.t(x) = \lambda^* x.t'(x)$ for some terms $t(x), t'(x) \in T^{lin}_{\Sigma_A}(x)$. 
    Then $A$ satisfies $t=(\lambda ^*x.t)x = (\lambda^* x.t')x=t'$. 
\end{proof}

\begin{lemma}
    \label{lem:abstr}
    Let $A$ be a $\mathbf{BCI}$-algebra that satisfies the $\lambda^*$-closure of the identities of Table \ref{tab:linear}. 
    Let $t, t' \in T^{lin}_{\Sigma_A}(x_1, \ldots, x_n,x)$. 
    If $t=t'$ in $A_{lin}[x_1, \ldots, x_n, x]$, then $\lambda^* x. t = \lambda^* x. t'$ in $A_{lin}[x_1, \ldots, x_n]$. 
\end{lemma}

\begin{proof}
    Firstly, we observe that if $A$ satisfies the $\lambda^*$-closure of the identities of Table \ref{tab:linear}, then so does $A_{lin}[x_1, \ldots, x_n]$;  
    and that if $t=t'$ in $A_{lin}[x_1, \ldots, x_n, x]$, then the identity (in $x$) $t=t'$ holds in $A_{lin}[x_1, \ldots, x_n]$.
    We prove the statement by induction on the structure of the derivation that leads to $t=t'$. 
    If $t=t'$ is obtained by reflexivity, symmetry, or transitivity the thesis follows trivially. 
    The case when $t=t'$ is obtained by application is also easy to handle.   
    It remains to consider the case where $t=t'$ is an instance of one of the identities defining $\mathbf{BCI}$-algebras (it cannot be an instance of the $\lambda^*$-closure of the equations of Table \ref{tab:linear} since they do not contain free variables).
    According to which argument of the axiom contains $x$, we unfold $\lambda^*x$ on both sides and conclude with
the corresponding identity of Table~\ref{tab:linear} (via Lemma~\ref{lem:closure}), the three positions of $x$ in
the axioms for $\bb$ and for $\cc$ of the last two rows of Table~\ref{tab:linear}.
We illustrate the method on one case.
    We assume that $t=t'$ is $\bb uvw = u(vw)$ with $u,v\in T^{lin}_{\Sigma_A}(x_1, \ldots, x_n)$ and $w \in T^{lin}_{\Sigma_A}(x_1, \ldots, x_n,x)$. 
    Then in $A_{lin}[x_1, \ldots, x_n]$: 
    $$\lambda^* x. \bb uvw = \bb(\bb uv)(\lambda^* x.w) = \bb u(\bb v (\lambda^* x.w))=\bb u(\lambda^* x.vw) =\lambda^*x.u(vw)\text{,}$$
    where, in the second equality, we used that $A_{lin}[x_1, \ldots, x_n]$ satisfies $\bb(\bb xy)z=\bb x(\bb yz)$ by Lemma \ref{lem:closure}. 
    The other cases are treated similarly. 
\end{proof}

Observe that if $t(x_1, \ldots, x_n) = t'(x_1, \ldots, x_n)$ in $A_{lin}[x_1, \ldots, x_n]$, then $A$ satisfies the equation $t=t'$ but the converse is not true. 

\begin{theorem}
    \label{thm:equiv}
    A $\mathbf{BCI}$-algebra $A$ is a linear $\lambda$-algebra iff $A$ satisfies the $\lambda^*$-closure of the identities of Table \ref{tab:linear}. 
\end{theorem}
\begin{proof}
    One direction follows from Lemma \ref{lem:table}. 
    Conversely, to prove \eqref{eq:bc-derived}, observe first of all that for every $a \in A$, 
    \begin{equation}
        \label{eq:mistery}
        \lambda^* x. ax = \bb a \ii = \bb \ii a = 1 a\text{,} 
    \end{equation}
    where in the second-to-last last equality we used $\bb \ii x = \bb x \ii$. Then by \eqref{eq:mistery}, by $1(\bb xy)=\bb xy$, $1(\bb x)=\bb x$ and $1\bb=\bb$:  
    \begin{align*}
        \lambda^* xyz. x(yz) = \lambda^* xyz. \bb xyz = \lambda^* xy. 1(\bb xy) = \lambda^* xy. \bb xy = \lambda^* x. 1(\bb x) = \lambda^* x. \bb x = 1 \bb =\bb
    \end{align*}
    and the other identity is proved similarly. 
    We have to show \eqref{eq:def-linear}.
    We have already proved that if $M \to_{\beta} N$, then $\llbracket M \rrbracket = \llbracket N \rrbracket$. 
    Moreover, it is easy to see that if $M=_{\beta} M'$ and $N=_{\beta} N'$, then $\llbracket MN \rrbracket = \llbracket M'N' \rrbracket$
    for all $M', N' \in \Lambda_{lin}$. 
    It remains to prove that if $M=_{\beta} N$, then $\llbracket \lambda x.M \rrbracket = \llbracket \lambda x.N \rrbracket$. 
    In light of the definition of $\llbracket - \rrbracket$, this follows from Lemma \ref{lem:abstr}. 
\end{proof}

\begin{corollary}
    \label{cor:subvariety}
    The class of linear $\lambda$-algebras is a subvariety of the variety of $\mathbf{BCI}$-algebras. 
\end{corollary}

The explicit list of equations axiomatising linear $\lambda$-algebras can be found in Appendix \ref{app:linear}. 
Hasegawa's axioms for extensional $\mathbf{BCI}$-algebras~\cite[Section 5.1]{H23} axiomatise extensional linear $\lambda$-algebras: 
a $\mathbf{BCI}$-algebra satisfies Hasegawa's axioms if and only if it is a linear $\lambda$-algebra that satisfies 
$\bb \ii=\ii$, which is the linear counterpart of the $\eta$-rule. 
We leave the verification of this fact for future work. 

Let $\mathbf{LinAlg}$ be the category whose objects are linear $\lambda$-algebras and whose morphisms are, in light of Corollary \ref{cor:subvariety}, homomorphisms of $\mathbf{BCI}$-algebras. 

\begin{remark}
    \label{rem:smart}
    If $A$ is a linear $\lambda$-algebra, there is, by virtue of Lemma \ref{lem:abstr} and Theorem \ref{thm:equiv}, a well-defined function $\lambda^* x.(-) : A_{lin}[x_1, \ldots, x_n,x] \to A_{lin}[x_1, \ldots, x_n]$. 
    If $f: A \to B$ is a homomorphism of linear $\lambda$-algebras, there is an obvious induced homomorphism $f: A_{lin}[x_1, \ldots, x_n] \to B_{lin}[x_1, \ldots, x_n]$ and for every $t \in A_{lin}[x_1, \ldots, x_n,x]$, 
    $f(\lambda^* x.t) = \lambda^* x. f(t)$. 
\end{remark}

\subsection{A categorical characterisation}
\label{sec:cat}
Recall that $\mathbf{LinAlg}$ is the category of linear $\lambda$-algebras.
We define a functor $F: \mathbf{LinAlg} \to \alg(L)$ and prove that it is part of an equivalence.  
Given a linear $\lambda$-algebra $A$, we define $F(A)$ to be the $L$-algebra with underlying set $A$ and action 
\begin{equation*}
    \bar{M} \cdot (a_1, \ldots, a_n) = \llbracket M \rrbracket [a_1 /x_1 , \ldots, a_n/x_n]
\end{equation*}
for every linear $\lambda$-term $M$ and its class $\bar{M} \in L(n)$.
The action is well defined by condition \eqref{eq:def-linear} of the definition of linear $\lambda$-algebra. 
Observe that $\llbracket M \rrbracket \in A_{lin}[x_1, \ldots, x_n]$ and that $\llbracket M \rrbracket [a_1 /x_1 , \ldots, a_n/x_n] \in A$. 
To prove that $F(A)$ is an $L$-algebra, one has to verify the conditions of Definition \ref{def:algebra} and this is done in Appendix \ref{app:proofs}. 

Recall the functor $G: \alg(L) \to \mathbf{BCI}$. 
\begin{lemma}
    \label{lem:functorf}
    If $f: A \to B$ is a linear $\lambda$-algebra homomorphism, then $f: F(A) \to F(B)$ is a morphism of $L$-algebras. 
\end{lemma}
\begin{proof}
    Firstly, we prove the following claim. 
    \begin{claim}
    \label{cl:claim}
    If $A$ is a linear $\lambda$-algebra, then the identity function $A \to GF(A)$ is a homomorphism of $\mathbf{BCI}$-algebras. 
    \end{claim}
    \begin{proof}
    For every $a_1, a_2 \in A$, $\app \cdot (a_1, a_2) = \llbracket x_1 x_2 \rrbracket [a_1/x_1, a_2/x_2] = a_1 a_2$. 
    Moreover, $\bb=\llbracket \mathbf{B} \rrbracket$ and $\cc=\llbracket \mathbf{C} \rrbracket$ by \eqref{eq:bc-derived} and $\ii=\llbracket \mathbf{I} \rrbracket$ by the definition of $\llbracket - \rrbracket$. 
    \end{proof}
    Now, since the identities $A \to GF(A), B \to GF(B)$ are homomorphisms of $\mathbf{BCI}$-algebras, $f : GF(A) \to GF(B)$ is a homomorphism of $\mathbf{BCI}$-algebras. 
    This is a morphism of $L$-algebras $F(A) \to F(B)$ by the fullness of $G$ (Proposition \ref{prop:ff}). 
\end{proof}

Thanks to the results of the previous sections we can prove the following. 

\begin{lemma}
    For every $A \in \alg(L)$, $G(A) \in \mathbf{LinAlg}$. 
\end{lemma}
\begin{proof}
By Theorem \ref{thm:equiv} it suffices to check that $G(A)$ satisfies the
$\lambda^*$-closures of the identities of Table~\ref{tab:linear}. 
These are closed equations $t=t'$ in $\bb,\cc,\ii$. 
For a closed term $t$, let $M_t \in\Lambda_{lin}(0)$
be obtained by replacing $\bb,\cc,\ii$ with $\mathbf{B},\mathbf{C},\mathbf{I}$. 
By induction on $t$, using the definition
of $G(A)$, we have that the interpretation of $t$ in $G(A)$ is $M_t \cdot()$. 
Now the $\lambda^*$-closure $t=t'$ of each identity
of Table~\ref{tab:linear} holds as a $\beta$-equivalence $M_t=_\beta M_{t'}$ and therefore $t=t'$ holds in $G(A)$. 
\end{proof}

Therefore, by abuse, we indicate with $G$ also the corestriction $G: \alg(L) \to \mathbf{LinAlg}$ of the functor $G: \alg(L) \to \mathbf{BCI}$. 

\begin{proposition}
    \label{prop:equivalence}
    The pair of functors 
    $F: \mathbf{LinAlg} \to \alg(L)$ and $G: \alg(L) \to \mathbf{LinAlg}$
    is an adjoint equivalence, concrete over the category of sets. 
\end{proposition}

\begin{proof}
    Let $A$ be a linear $\lambda$-algebra. 
    Claim \ref{cl:claim} says the identity function $A \to GF(A)$ is a homomorphism of $\mathbf{BCI}$-algebras; 
    this is enough by Corollary \ref{cor:subvariety} for showing that it is a morphism in $\mathbf{LinAlg}$.  
    Conversely, let $A$ be an $L$-algebra with action $\cdot :L(n) \times A^n \to A$. 
    We need to show that the identity function $FG(A) \to A$ is a morphism of $L$-algebras. 
    Let $M \in \mathbf{BCI} \langle x_1, \ldots, x_n \rangle$ and write $\bar{M}$ for the equivalence class $x_1, \ldots, x_n \vdash M \in L(n)$. 
    We show that for every $x_1, \ldots, x_n \vdash M \in L(n)$ and every $a_1, \ldots, a_n \in A$ 
    $$ \llbracket M \rrbracket [a_1/x_1, \ldots, a_n/x_n] = \bar{M} \cdot (a_1,\ldots, a_n)$$
    by induction on $M$. 
    If $n=1$ and $M=x_1$, then $ \llbracket x_1 \rrbracket [a_1/x_1] = a_1 = x_1 \cdot a_1$ for all $a_1 \in A$. 
    If $M=\mathbf{B}, \mathbf{C}$ or $\mathbf{I}$, then 
    $$ \llbracket \mathbf{B} \rrbracket = \bb = \mathbf{B} \cdot (), \quad \llbracket \mathbf{C} \rrbracket = \cc = \mathbf{C} \cdot (), \quad \llbracket \mathbf{I} \rrbracket = \ii = \mathbf{I} \cdot ()\text{.} $$
    Now, let $M=PQ$; then $\bar{M}= \app(\bar{P}, \bar{Q})$ for some $\bar{P} \in L(i)$, $\bar{Q} \in L(j)$ with $i +j=n$ (for simplicity we assume that there is no permutation of the variables).
    We have 
    \begin{align*}
        \llbracket M \rrbracket [a_1/x_1, \ldots, a_n/x_n] & = \llbracket PQ \rrbracket [a_1/x_1, \ldots, a_n/x_n] \\
        & = \llbracket PQ [a_1/x_1, \ldots, a_n/x_n]  \rrbracket  \\
        & = \llbracket (P [a_1/x_1, \ldots, a_i/x_i])( Q [a_{i+1}/x_{i+1}, \ldots, a_n/x_n]  )\rrbracket \\
        & = (\llbracket P \rrbracket [a_1/x_1, \ldots, a_i/x_i]) (\llbracket Q \rrbracket [a_{i+1}/x_{i+1}, \ldots, a_n/x_n]) \\
        & = (\bar{P} \cdot (a_1, \ldots, a_i)) (\bar{Q} \cdot (a_{i+1}, \ldots, a_n))  \\
        & = \bar{M} \cdot (a_1,\ldots, a_n) \text{,}
    \end{align*}
    where the second and fourth equalities follow from Remark \ref{rem:comm} and the fifth equality from the inductive assumption. 
    Finally, Lemma \ref{lem:hindley} shows that every element of $L(n)$ is the class $\bar{M}$ of some $M \in \mathbf{BCI} \langle x_1, \ldots, x_n \rangle$. 
\end{proof}

\section{Semiclosed operads}
The operad of linear $\lambda$-terms $L$ enjoys a very important property, that we set out to study in this section. 

Let $O$ be an operad. 
Let $\partial O$ be the sequence of sets given by $(\partial O)(n):=O(n+1)$\footnote{In the language of species, $\partial$ is the
derivative, see e.g. \cite{F05,FGHW08}.}. 
There are maps
\begin{equation*}
   (\partial O)(n)\times O(k_1)\times\cdots\times O(k_n)\to(\partial O)(k_1+\cdots+k_n), \,
 (f,g_1,\dots,g_n)\mapsto f\circ(g_1,\dots,g_n,1) 
\end{equation*}
which satisfy the axioms of Definition \ref{def:module}. 
Moreover, $\partial O$ inherits an action of $S_n \subseteq S_{n+1}$ from $O$. 
Therefore, $\partial O$ is a module over $O$. 
Recall that $O$ is a module over itself. 

\begin{definition}[cf.\ {\cite[Definition 3.1]{H17}}]
 An operad $O$ is \emph{semiclosed} if it is equipped with morphisms of $O$-modules
 $ r: O\to\partial O$ and $s: \partial O\to O$
 such that $r\circ s=\mathrm{id}_{\partial O}$, i.e.\ $\partial O\triangleleft O$
 is a retraction in the category of $O$-modules. 
 A morphism of semiclosed operads is a morphism of operads $\phi\colon O\to O'$
 such that $\phi\circ r=r'\circ\phi$ and $\phi\circ s=s'\circ\phi$.
\end{definition}

 Let $r_n\colon O(n)\to O(n+1)$ and $s_n\colon O(n+1)\to O(n)$ be maps with
 $r_n \circ s_n=\mathrm{id}_{O(n+1)}$. Then $r=(r_n)_n$ and $s=(s_n)_n$ are morphisms of 
 modules iff $r,s$ satisfy the conditions (3) and (4) of Definition \ref{def:module} and for all $f\in O(n)$, $g\in O(n+1)$, $h_i\in O(k_i)$,
 \begin{equation}
  \label{eq:comp1}
  r(f\circ(h_1,\dots,h_n))=r(f)\circ(h_1,\dots,h_n,1) 
 \end{equation}
 and 
  \begin{equation}
  \label{eq:comp2}
    s(g\circ(h_1,\dots,h_n,1))=s(g)\circ(h_1,\dots,h_n)\text{.}
 \end{equation}

\begin{lemma}
    The operad $L$ is semiclosed. 
\end{lemma}
\begin{proof}
        We define two maps $r_L: L(n) \to L(n+1)$ and $s_L: L(n+1) \to L(n)$:  
    \begin{equation*}
        r_L : x_1, \ldots, x_{n} \vdash M \quad \mapsto \quad Mx_{n+1} \quad \text{and} \quad  s_L : x_1, \ldots, x_n, x_{n+1} \vdash M \quad \mapsto \quad \lambda x_{n+1}. M\text{.} 
    \end{equation*}
    The retraction $r_L \circ s_L = \id$ follows from $\beta$-equivalence: $(\lambda x_{n+1}. M)x_{n+1} =_{\beta} M$. 
    Condition \eqref{eq:comp1} is obvious and \eqref{eq:comp2} is the remark that, if $x_1, \ldots, x_n,x \vdash M$, then 
    $(\lambda x. M) [N_1/x_1, \ldots, N_n/ x_n] = \lambda x. M[N_1/x_1, \ldots, N_n/x_n]$. 
\end{proof}

Let $O$ be a semiclosed operad with retraction $r$ and section $s$. 
We define a function $\iota: L \to O$ as follows. 
Let $M \in L(n)$; we define $\iota(M)$ by induction on $M$. 
If $n=1$ and $M=x$, $\iota(x)=1 \in O(1)$. 
Let $\app \in O(2)$ be the image of $1 \in O(1)$ through the map $r: O(1) \to O(2)$. 
If $M=PQ$, we define $\iota(PQ):= \app(\iota(P), \iota(Q))$ and if $M=\lambda x.N$, $\iota(M):=s(\iota(N))$. 
We now prove that if $(\lambda x. M) N \to_{\beta} M[N/x]$, then $\iota((\lambda x. M) N) = \iota(M[N/x])$. 
To show this we need the following lemma. 

\begin{lemma}
    \label{lem:subst}
    For every $x_1, \ldots, x_n \vdash M$ in $L(n)$, $\iota(M[N/x_n]) = \iota(M)(1, \ldots, 1,\iota(N))$. 
\end{lemma}
\begin{proof}
    The proof is by induction on the structure of $M$. 
    If $n=1$ and $M=x_1$, then the claim is trivially true. 
    If $M=PQ$, with, say, $x_n$ free in $Q$, then 
    \begin{align*}
        \iota(M[N/x_n]) & = \iota(PQ[N/x_n]) \\
        & = \app(\iota(P), \iota(Q[N/x_n])) & \text{by definition of $\iota$}\\
        & = \app(\iota(P),\iota(Q)(1, \ldots, 1,\iota(N))) & \text{by induction}\\
        & =  \app(\iota(P), \iota(Q))(1, \ldots, 1,\iota(N)) & \text{by Definition \ref{def:operad}(3)}\\
        & = \iota(M)(1, \ldots, 1,\iota(N)) & \text{by definition of $\iota$}
    \end{align*}
    If $M=\lambda x. M'$, then 
    \begin{align*}
        \iota(M[N/x_n]) & = \iota((\lambda x. M')[N/x_n]) \\
        & = \iota(\lambda x. M'[N/x_n]) \\ 
        & = s(\iota(M')[N/x_n]) & \text{by definition of $\iota$} \\
        & = s(\iota(M')(1, \ldots, 1,\iota(N))) & \text{by induction} \\
        & = s(\iota(M'))(1, \ldots, 1,\iota(N)) & \text{by \eqref{eq:comp2}} \\
        & = \iota(M)(1, \ldots, 1,\iota(N)) & \text{by definition of $\iota$.}  
    \end{align*}
    This proves the desired claim.
\end{proof}

Coming back to the proof that if $(\lambda x. M) N \to_{\beta} M[N/x]$, then $\iota((\lambda x. M) N) = \iota(M[N/x])$, observe that 
\begin{align*}
    \iota(M) & = r(s(\iota(M))) & \text{because $r \circ s= \id$} \\
    & = r(1)(s(\iota(M)),1) & \text{by \eqref{eq:comp1}} \\
    & = \app(s(\iota(M)),1) & \text{by definition of $\app$}
\end{align*}
so that 
$\iota(M)(1, \ldots, 1, \iota(N)) = \app(s(\iota(M)), \iota(N))\text{.}$ Therefore, by Lemma \ref{lem:subst}
$$\iota((\lambda x. M) N) = \app(s(\iota(M)), \iota(N)) = \iota(M)(1, \ldots, 1, \iota(N)) = \iota(M[N/x])\text{.}$$

\begin{proposition} [cf. {\cite[Theorem 3.8]{H17}}]
    \label{prop:initial}
   The operad $L$ is the initial semiclosed operad.  
\end{proposition}
\begin{proof}
    Let $O$ be a semiclosed operad with retraction $r$ and section $s$.  
    We have just proved that there is a well-defined map $\iota : L \to O$, i.e., one that is invariant under $=_{\beta}$. 
    The fact that $\iota$ is a morphism of operads follows from a simple generalisation of Lemma \ref{lem:subst}. 
    Moreover, $\iota \circ r_L = r \circ \iota$: if $x_1, \ldots, x_n \vdash M \in L(n)$, 
        $\iota(r_L(M)) = \iota(Mx_{n+1}) = \app(\iota(M), 1) = r(1)(\iota(M),1) = r(\iota(M)) \text{.}$
    The proof that $\iota \circ s_L = s \circ \iota$ is similar. 
    Finally, if $\phi: L \to O$ is a morphism of semiclosed operads, then $r(\phi(M)) = \phi(xy) (\phi(M),1)$ for every $x_1, \ldots, x_n \vdash M$. 
    When $M=x$, this implies that $\phi(x_1 x_2)=r(1)$, i.e., $\phi(x_1 x_2)=\app$. 
    Moreover, $s(\phi(M))=\phi(\lambda x_{n+1} . M)$. 
    Therefore, inductively, for every linear $\lambda$-term $M$, $\phi(M)=\iota(M)$. 
    This proves that $\iota : L \to O$ is the unique morphism of semiclosed operads between $L$ and $O$, i.e., $L$ is initial. 
\end{proof}

\subsection{Scott's representation theorem for the linear lambda calculus}
Let $O$ be an operad, and let $\cat O$ be the associated PROP. 
Day's closed structure on the category $\psh(\cat O)$ (see Section \ref{sec:monoidal}) 
gives a symmetric closed monoidal category $(\psh  (\cat O), \otimes,  \multimap, J)$.
Recall that the exponentials are given by 
\begin{equation}
    (F \multimap G)(c) = \int_{x \in \cat C} \mathbf{Set}(F(x),G(c \otimes x))
\end{equation}
for every $c \in \cat C$. 
Specialising this formula to $F= \cat C(-,a)$, we obtain $  (\cat C(-,a) \multimap G)(c) \simeq G(c \otimes a)$. 
Specialising in addition to $G=\cat C (-,b)$ we get 
\begin{equation}
    \label{eq:multi}
     (\cat C(-,a) \multimap \cat C(-,b))(c) \simeq \cat C(c \otimes a,b). 
\end{equation}

\begin{theorem}
    \label{thm:semi-linear}
    Let $O$ be a semiclosed operad and let $\cat O$ be its associated \emph{PROP}. 
    Then $\cat O(-,1)$ is a linear reflexive object in the closed monoidal category $\psh (\cat O)$.
\end{theorem}
\begin{proof}
    Applying \eqref{eq:multi} to $\cat C = \cat O$, $a=b=1$ and $c=n$ we obtain: 
    \begin{align*}
        (\cat O(-,1) \multimap \cat O(-,1))(n) & = \int_{x } \mathbf{Set} (\cat O(x,1), \cat O (x+n,1)) \\
        & \simeq \psh (\cat O) (\cat O(-,1), \cat O (-+n,1)) \\
        & \simeq \cat O(n+1,1). 
    \end{align*}
As $O$ is semiclosed, for every $n \in \mathbb N$, $(\cat O(-,1) \multimap \cat O(-,1))(n) = O(n+1) \triangleleft O(n) = \cat O (n,1)$. 
\end{proof}

\begin{remark}
    \label{rem:linear-semi}
    Let $ (\cat C, \otimes, \multimap, I)$ be a symmetric closed monoidal category. 
Let $x \in \cat C$ be a linear reflexive object. 
    Then the endomorphism operad (see Example \ref{ex:endo}) $O$ on $x$ is semiclosed: 
\begin{equation*}
    O(n+1)  =  \cat C(x^{\otimes n} \otimes x, x) 
    \simeq  \cat C(x^{\otimes n} , x \multimap x) 
    \triangleleft \cat C(x^{\otimes n}, x) 
    = O(n).
\end{equation*}
\end{remark}

Putting Lemma \ref{lem:endo}, Remark \ref{rem:linear-semi} and Theorem \ref{thm:semi-linear} together, we obtain the following. 
\begin{theorem} [cf. {\cite[Theorem 3.10]{H17}}]
    \label{thm:scott}
    Let $O$ be an operad, with associated \emph{PROP} $\cat O$. 
    Then $O$ is isomorphic to the endomorphism operad of $\cat O(-,1)$ in the monoidal closed category $\psh  (\cat O)$.
    Moreover, $O$ is semiclosed iff $\cat O(-,1)$ is linear reflexive in $\psh  (\cat O)$.
\end{theorem}

Thus semiclosed operads correspond to reflexive objects in closed monoidal categories. 
Since semiclosed operads are one of the possible equivalent ways to describe a model of the linear $\lambda$-calculus (see the next section), 
the last result is Scott's representation theorem for the linear $\lambda$-calculus.

\section{A triple equivalence}

We have seen that the algebras for the operad $L$ can be equivalently described as a subvariety of $\mathbf{BCI}$-algebras. 
In this section we prove that these two categories are in fact equivalent to the category of semiclosed operads. 
The argument proceeds as follows: 
\begin{itemize}
    \item we define a functor from semiclosed operads to $\alg(L)$;
    \item we define a functor from $\mathbf{LinAlg}$ to the category of semiclosed operads; 
    \item we prove that they are inverses modulo the equivalence of Proposition \ref{prop:equivalence}.  
\end{itemize}

Firstly, we associate an $L$-algebra to any semiclosed operad (this construction is similar to Hyland's--see \cite[Proposition 4.3]{H17}). 
Given a semiclosed operad $C$, 
$C(0)$ is a $C$-algebra with action 
    $\cdot : C(n) \times C(0)^n \to C(0)$ given by $\circ: C(n) \times C(0)^n \to C(0)\text{.}$
Since $L$ is the initial semiclosed operad by Proposition \ref{prop:initial}, there is a unique morphism of operads $\iota: L \to C$. 
We obtain an action $\cdot : L(n) \times C(0)^n \to C(0)$ by 
\begin{equation*}
    L(n) \times C(0)^n \xrightarrow{\iota_n \times \id} C(n) \times C(0)^n \xrightarrow{\cdot} C(0)\text{.}
\end{equation*}

\begin{proposition}
    The map $C \mapsto C(0)$ is a functor from the category of semiclosed operads to the category of $L$-algebras. 
\end{proposition}
\begin{proof}
    If $\phi: C \to D$ is a morphism of semiclosed operads, then, as the following diagram is commutative,
    \begin{equation*}
        \begin{tikzcd}
            L(n) \times C(0)^n \arrow[r, "\iota \times \id"] \arrow[d, "\id \times \phi^n"] & C(n)\times C(0)^n \arrow[r, "\circ"] \arrow[d, "\phi \times \phi^n"] & C(0) \arrow[d, "\phi"] \\
            L(n) \times D(0)^n \arrow[r, "\iota \times \id"]  & D(n)\times D(0)^n \arrow[r, "\circ"]  & D(0) 
        \end{tikzcd}
    \end{equation*}
    the map $\phi_0 : C(0) \to D(0)$ is a morphism of $L$-algebras. 
\end{proof}

We now associate to every linear $\lambda$-algebra $A$ a semiclosed operad $O_A$. 
For every $n \in \mathbb{N}$, let $O_A(n):=\{a \in A : \mathbf{1}_n a = a\}$, where by abuse we employ the notation 
    $\mathbf{1}_n := \lambda^* x x_1 \ldots x_n. x x_1 \cdots x_n \in A$
to denote the interpretation of the combinator $\mathbf{1}_n$ in $A$. 
Observe that if $n=0$, then $O_A(0)=\{a \in A : \ii a =a\}=A$. 
We endow the sequence $\{O_A(n) : n \in \mathbb{N}\}$ with the following structure: 
\begin{itemize}
    \item we define the identity to be $\ii \in O_A(1)$, recalling that in a linear $\lambda$-algebra $1 \ii = \ii$ holds; 
    \item we define a composition $\circ: O_A(n) \times O_A(k_1) \times \cdots \times O_A(k_n) \to O_A(k_1 + \cdots + k_n)$ by 
    $$a \circ (a_1, \ldots, a_n) := \lambda^* x^1_1\ldots x^n_{k_n}. a (a_1x^1_1\cdots x^1_{k_1}) \cdots (a_n x^n_1 \cdots x^n_{k_n}) \text{;}$$
    \item for every $a \in A$ and $\sigma \in S_n$, we define 
    $a \sigma := \lambda^* x_1 \ldots x_n . a x_{\sigma 1} \cdots x_{\sigma n} \text{.}$
\end{itemize}

First of all observe that the composition and the action of $S_n$ are well-defined as $\mathbf{1}_n(a \sigma) = a \sigma$ and 
$\mathbf{1}_{k_1 + \cdots + k_n}(a\circ(a_1, \ldots, a_n))=a \circ (a_1, \ldots, a_n)\text{.}$

\begin{lemma}
    \label{lem:oa-operad}
    The above structure makes the sequence of sets $O_A:=\{O_A(n) : n \in \mathbb{N}\}$ an operad. 
    Moreover, the operad $O_A$ is semiclosed. 
\end{lemma}
\begin{proof}
    The verification of the conditions of Definition \ref{def:operad} is addressed in Appendix \ref{app:proofs}. 
    To show that $O_A$ is semiclosed, we define two maps $s: O_A(n+1) \to O_A(n)$ and $r: O_A(n) \to O_A(n+1)$. 
    Since $\mathbf{1}_{n}(\mathbf{1}_{n+1}a)=a$ for every $a \in O_A(n+1)$, we can define $s$ as the inclusion $O_A(n+1) \subseteq O_A(n)$. 
    Moreover, we define $r(a):=\mathbf{1}_{n+1} a$ for every $a \in O_A(n)$. 
    The map $r$ is well-defined as $\mathbf{1}_{n+1}$ is idempotent. 
    Moreover, $r \circ s= \id$ since $r(s(a))=r(a)=\mathbf{1}_{n+1} a=a$ for every $a \in O_A(n+1)$. 
    Finally, the verification of conditions \eqref{eq:comp1} and \eqref{eq:comp2} is straightforward. 
\end{proof}

\begin{proposition}
    The map $A \mapsto O_A$ is a functor from the category of linear $\lambda$-algebras to the category of semiclosed operads. 
\end{proposition}
\begin{proof}
    Let $f:A \to B$ be a homomorphism of linear $\lambda$-algebras. 
    We define a map $\phi: O_A \to O_B$ by letting, for each $n \in \mathbb{N}$, $\phi_n : O_A(n) \to O_B(n)$ be given by $\phi_n(a) = f(a)$. 
    We show that $\phi$ is a morphism of semiclosed operads. 
    First of all, $\phi_n$ is a well-defined function as $\mathbf{1}_{n} f(a)=f(\mathbf{1}_n a)=f(a)$. 
    Moreover,  
    for every $a \in O_A(n)$ and $a_i \in O_A(k_i)$ for $i=1, \ldots, n$: 
    \begin{align*}
        \phi_n (a(a_1, \ldots, a_n)) & = \phi_n(\lambda^* \vec{x}. a (a_1 \vec{x}_1) \cdots (a_n \vec{x}_n)) & \text{definition of $\circ$} \\
        & = f(\lambda^* \vec{x}. a (a_1 \vec{x}_1) \cdots (a_n \vec{x}_n)) & \text{definition of $\phi$} \\
        & = \lambda^* \vec{x}. f(a) (f(a_1)\vec{x}_1) \cdots (f(a_n) \vec{x}_n) & \text{$f$ is homomorphism} \\
        & = \lambda^* \vec{x}. \phi_n(a) (\phi_{k_1}(a_1)\vec{x}_1) \cdots (\phi_{k_n}(a_n) \vec{x}_n) & \text{definition of $\phi$} \\
        & =  \phi_n (a) (\phi_{k_1} (a_1), \ldots, \phi_{k_n} (a_n)) & \text{definition of $\circ$.} 
    \end{align*}
   Finally, the map $\phi$ respects the semiclosed structure. 
   Let $r_A, s_A$ be the retraction and the section of $O_A$ and $r_B, s_B$ be the retraction and the section of $O_B$, respectively. 
   Then, for every $a \in O_A(n)$, we see that 
    $\phi(r_A(a))=f(\mathbf{1}_{n+1}a)=\mathbf{1}_{n+1}f(a)=r_B(\phi(a))$
   and similarly we see that $\phi \circ s_A = s_B \circ \phi$. 
\end{proof}

Recall the functor $F$ defined at the beginning of Section \ref{sec:cat}. 

\begin{proposition}
    \label{prop:oaisof}
     Let $A$ be a linear $\lambda$-algebra. Then $O_A(0) \simeq F(A)$ as $L$-algebras. 
\end{proposition}

\begin{proof}
   First of all observe that at the level of the underlying sets $O_A(0)=A$ and recall that the action by $L$ on $O_A(0)$ is given by 
   \begin{equation*}
        \begin{tikzcd}
        L(n) \times O_A(0)^n \arrow[r, "\iota \times \id"] & O_A(n)  \times O_A(0)^n \arrow[r, "\circ"] & O_A(0) \text{.} 
    \end{tikzcd}
   \end{equation*}
   Since $G$ is fully faithful (Proposition \ref{prop:ff}), if two $L$-algebras $B$ and $B'$ have the same underlying set and $G(B) =G(B')$, 
   then the identity $G(B) \to G(B')$ is a homomorphism of $\mathbf{BCI}$-algebras and so the identity $B \to B'$ is a morphism of $L$-algebras. 
   To show that the identity $O_A(0) \to F(A)$ is a morphism of $L$-algebras, it is therefore enough to show that $G(O_A(0)) \to GF(A)$, or equivalently $G(O_A(0)) \to A$ by Claim \ref{cl:claim}, 
   is a homomorphism of $\mathbf{BCI}$-algebras. 
   Observe that $\app = \mathbf{1} \in O_A(2)$, so that $\app(a,b) = \mathbf{1}ab = ab$ for all $a,b \in A$. 
   Moreover, 
   \begin{align*}
    \iota(\mathbf{B}) & = s(s(s(\app(1,\app(1,1))))) & \text{by definition of $\iota$} \\
    & = s(s(s(\app(\ii,\app(\ii,\ii))))) & \text{by definition of $1$} \\
    & = \app(\ii,\app(\ii,\ii)) & \text{by definition of $s$} \\
    & = \lambda^* x_1 x_2 x_3 . \mathbf{1}(\ii x_1) (\mathbf{1}(\ii x_2)(\ii x_3)) & \text{by definition of $\circ$} \\
    & = \lambda^* x_1 x_2 x_3 . x_1(x_2 x_3) \\
    & = \bb & \text{by \eqref{eq:bc-derived}}
   \end{align*}
   and analogously for $\mathbf{C}$. 
   Finally, $\iota(\mathbf{I})= s(1)= s(\ii)=\ii$, proving the result. 
\end{proof}

\begin{proposition}
 \label{prop:oaisoc}
 Let $C$ be a semiclosed operad and let $A$ be the linear $\lambda$-algebra $A:= G(C(0))$.
 Then $O_{A} \simeq C$ in the category of semiclosed operads. 
\end{proposition}

The proof unravels through different steps. 
The idea is that an element of $C(n)$ is an element of $C(0)=A$ with $n$ arguments abstracted. 
We define two morphisms of semiclosed operads $\phi: C \to O_{A}$ and $\psi: O_{A} \to C$ inverses to each other. 
First, we define two functions $\phi_n: C(n) \to A=C(0)$ and $\psi_n: A=C(0) \to C(n)$ by induction on $n \in \mathbb{N}$.  
If $n=0$, $\phi_0$ and $\psi_0$ are the identity of $C(0)$.  
Since $C$ is semiclosed, there is a pair of maps $s_n : C(n+1) \to C(n)$ and $r_n : C(n) \to C(n+1)$ such that $r_n \circ s_n = \id_{C(n+1)}$. 
Assuming $n \ge 0$, we let 
\begin{equation}
    \label{eq:def-phi-psi}
    \phi_{n+1}(f) = (\phi_n \circ s_n)(f) \quad \text{and} \quad \psi_{n+1}(a) = (r_n \circ \psi_n)(a)\text{.}
\end{equation}
Now, observe that the image of the function $\phi_{n}$ is $O_A(n)$; 
moreover, we shall look at the restriction of $\psi_n$ to $O_A(n) \subseteq A$. 
Hence, by abuse of notation, we denote by $\phi$ and $\psi$ also the maps $\phi_{n}: C(n) \to O_A(n)$ and $\psi_n: O_A(n) \to C(n)$. 
Inductively, we have: 
$$\psi_{n+1} \circ \phi_{n+1} = (r_n \circ \psi_n) \circ (\phi_n \circ s_n) = r_n \circ (\psi_n \circ \phi_n) \circ s_n = r_n \circ  s_n = \id_{C(n+1)}\text{.}$$
Now, recall that by Proposition \ref{prop:initial} we have a unique morphism $\iota : L \to C$. 
A simple induction on $n \ge 1$ proves the following lemma. 
\begin{lemma}
    \label{lem:last}
    Let $M$ be a linear $\lambda$-term. Then: 
    \begin{enumerate}
        \item $\iota(Mx_1 \cdots x_n) = (r_{n-1} \circ \cdots \circ r_0)(\iota(M))$; 
        \item $\iota(\lambda x_1 \ldots x_n. M) = (s_0 \circ \ldots \circ s_{n-1})(\iota(M))$ 
    \end{enumerate}
    and, consequently, 
\begin{enumerate}
    \setcounter{enumi}{2}
    \item $\iota(\lambda x_1 \ldots x_n. M x_1 \cdots x_n) = (s_0 \circ \ldots \circ s_{n-1} \circ r_{n-1} \circ \cdots \circ r_0)(\iota(M))$. 
\end{enumerate}
\end{lemma}

Therefore, by Lemma \ref{lem:last}, by the definition of the action on $C(0)$ and by the equivalence of Proposition \ref{prop:equivalence}, 
$(\phi_{n+1} \circ \psi_{n+1})(a) = (s_0 \circ \cdots \circ s_n \circ r_n \circ \cdots \circ r_0)(a) = \mathbf{1}_{n+1} a=a$ for every $a \in O_A(n+1)$.

\begin{lemma}
    \label{lem:tec}
    The maps $\phi$ and $\psi$ defined in \eqref{eq:def-phi-psi} are morphisms of semiclosed operads. 
\end{lemma}
\begin{proof}
The proof is technical and can be found in Appendix \ref{app:proofs}. 
It develops along the following lines. 
Compatibility with the semiclosed structure holds by construction, so the only real
content is compatibility with composition. 
We reduce it, using \eqref{eq:comp2} (the compatibility of $s$ with
composition), to composition with $1$ in the last position, and then to the general case by induction,
moving each argument into the last position by a transposition.
\end{proof}

This concludes the proof of Proposition \ref{prop:oaisoc}. 

Keeping Proposition \ref{prop:equivalence} in mind, as a consequence of Propositions \ref{prop:oaisoc} and \ref{prop:oaisof}, we have thus proved the following. 

\begin{theorem}
    \label{thm:final}
    The pair of functors $O_{(-)} \circ G$ and $(-)(0)$ is an adjoint equivalence between $\alg(L)$ and the category of semiclosed operads. 
\end{theorem}

\subparagraph*{Further work}
We conclude by pointing out some directions for future research.
The first is to extend the present framework to the $\lambda I$-calculus, whose semantics has recently attracted renewed interest \cite{Remy}. 
We aim to identify an operad-like structure capturing the abstract categorical semantics of the $\lambda I$-calculus (in which duplication of variables is allowed) in the spirit of Fiore, Plotkin, and Turi's approach to abstract syntax \cite{FPT} (see \cite{FR25} for a recent adaptation to substructural calculi)
and to use it as a foundation for studying $\lambda I$-algebras (based on the combinators $\mathbf{B}$, $\mathbf{C}$, $\mathbf{I}$, and $\mathbf{S}$), together with the corresponding notion of semiclosed structure. 
The category $\mathbf{Surj}$ of finite sets and surjective maps provides the appropriate combinatorial setting, playing for the $\lambda I$-calculus the same role that finite sets and bijections play for the linear one \cite{KNRS}.
The second direction is more conceptual. 
Jacobs' categorical treatment of simple type theory \cite{J92} and \cite[Section 2.5]{JacobsCLTT} suggests a close connection between the semantics of the typed and untyped $\lambda$-calculi. 
It is interesting to instantiate this view in the linear setting, which appears particularly promising since every linear $\lambda$-term admits a principal simple type, suggesting that the boundary between typed and untyped semantics is considerably thinner than in the classical Cartesian setting.
\enlargethispage{\baselineskip}

\bibliography{csl27-doi.bib}

\appendix 
\section{Missing proofs}
\label{app:proofs}

\subparagraph*{Proof of the claim at the beginning of Section \ref{sec:cat}.}
    We prove that the action defined makes $F(A)$ an $L$-algebra. 
    For the sake of an example, we check that, if $M \in L(n)$ and $N_i \in L(k_i)$ for $1 \le i \le n$, then 
\begin{align*}
    M \cdot (N_1, \ldots, N_n) \cdot (\vec{a}_1, \ldots, \vec{a}_n) & = \llbracket M[N_1/x_1, \ldots, N_n/x_n] \rrbracket[\vec{a}_1/\vec{y}_1, \ldots, \vec{a}_n/\vec{y}_n] \\
    & = \llbracket M \rrbracket [\llbracket N_1 \rrbracket/x_1, \ldots, \llbracket N_n \rrbracket/x_n][\vec{a}_1/\vec{y}_1, \ldots, \vec{a}_n/\vec{y}_n]        \\
    & = \llbracket M \rrbracket [\llbracket N_1 \rrbracket[\vec{a}_1/\vec{y}_1]/x_1, \ldots, \llbracket N_n \rrbracket[\vec{a}_n/\vec{y}_n]/x_n] \\
    & = \llbracket M \rrbracket [\llbracket N_1 \rrbracket \cdot \vec{a}_1 /x_1, \ldots, \llbracket N_n \rrbracket \cdot \vec{a}_n/x_n]\\ 
    & = M \cdot (N_1  \cdot \vec{a}_1, \ldots, N_n \cdot \vec{a}_n) 
\end{align*}
for every $(\vec{a}_1, \ldots, \vec{a}_n) \in A^{k_1 + \cdots + k_n}$, as desired. \qed

\subparagraph*{Missing part of the proof of Lemma \ref{lem:oa-operad}.}
We check some conditions:  
    \begin{align*}
        a \circ (\ii, \ldots, \ii) & = \lambda^* x_1 \ldots x_n. a(\ii x_1) \cdots (\ii x_n) \\
        & = \lambda^* x_1 \ldots x_n. ax_1 \cdots x_n \\
        & = \mathbf{1}_n a \\
        & = a;
    \end{align*}
    \begin{equation*}
        \ii \circ a = \lambda^* x_1 . \ii(ax_1) = \lambda^* x_1. ax_1 = \mathbf{1} a = a; 
    \end{equation*}
    \begin{align*}
        & a (a_1 (c^1_1, \ldots, c^1_{k_1}), \ldots, a_n(c^n_1, \ldots, c^n_{k_n})) \\
        & \hspace*{2cm} = a \circ (\lambda^* \vec{x}^1.a_1 (c^1_1 \vec{x}^1_1) \cdots (c^1_{k_1} \vec{x}^1_{k_1}), \ldots, \lambda^* \vec{x}^n.a_n (c^n_1 \vec{x}^n_1) \cdots (c^n_{k_n} \vec{x}^n_{k_n})) \\
        & \hspace*{2cm} = \lambda^* \vec{y}. a((a_1 (c^1_1 \vec{y}^1_1) \cdots (c^1_{k_1} \vec{y}^1_{k_1})) \cdots (a_n (c^n_1 \vec{y}^n_1) \cdots (c^n_{k_n} \vec{y}^n_{k_n})) )\\
        & \hspace*{2cm} = \lambda^* \vec{x}. a((a_1 \vec{x}_1) \cdots( a_n \vec{x}_n) ) \circ (c^1_1, \ldots, c^k_{k_n}) \\
        &  \hspace*{2cm} = a(a_1, \ldots, a_n) \circ (c^1_1, \ldots, c^k_{k_n}).
    \end{align*}
The others are similar.  \qed

\subparagraph*{Proof of Lemma \ref{lem:tec}.}
We show that $\phi$ is a morphism; 
that $\psi$ is a morphism can be shown with a similar computation. 
We leave it to the reader to prove that $\phi_1(1)=\ii$ and that $\phi_n(f \sigma)=\phi_n(f)\sigma$ for every $f \in C(n)$ and $\sigma \in S_n$.
We prove that 
\begin{equation}
    \label{eq:inductive}
    \phi_{k_1 + \cdots + k_n}(f(g_1, \ldots, g_n))=\phi_n(f)(\phi_{k_1}(g_1), \ldots, \phi_{k_n}(g_n))
\end{equation}
by induction on $n \in \mathbb{N}$. 
The base case is trivial. 
Now, assume \eqref{eq:inductive}. 
Before treating the general case, we claim that 
\begin{equation}
    \label{eq:claim2}
    \phi_{k_1 + \cdots + k_{n+1}}(f(g_1, \ldots, g_n,1))=\phi_{n+1}(f)(\phi_{k_1}(g_1), \ldots, \phi_{k_n}(g_n),1) \text{.}
\end{equation}
Indeed, 
\begin{align*}
    \phi_{k_1 + \cdots + k_{n+1}}(f(g_1, \ldots, g_n,1)) & = \phi_{k_1 + \cdots + k_{n}+1}(f(g_1, \ldots, g_n,1)) \\
    & = \phi_{k_1 + \cdots + k_{n}}(s_{k_1 + \cdots + k_{n}} f(g_1, \ldots, g_n,1)) & \text{by \eqref{eq:def-phi-psi}} \\
    & = \phi_{k_1 + \cdots + k_{n}}(s_n(f)(g_1, \ldots, g_n)) & \text{by \eqref{eq:comp2}} \\
    & = \phi_n(s_n(f))(\phi_{k_1}(g_1), \ldots, \phi_{k_n}(g_n)) & \text{by induction} \\
    & = \phi_{n+1}(f) (\phi_{k_1}(g_1), \ldots, \phi_{k_n}(g_n), 1) & \text{by \eqref{eq:def-phi-psi}}.
\end{align*}
Finally, let $\tau =(1,n+1) \in S_{n+1}$ and observe that $\tau^{-1}=\tau$.
Moreover, $\hat{\tau}$ is the permutation that exchanges the first and the last block. 
We have: 
\begin{align*}
    & \phi_{k_1 + \cdots + k_{n+1}} (f(g_1, \ldots, g_{n+1}))  \\
    & \hspace*{1cm} = \phi_{1+ \cdots + 1 +k_{n+1}} (f(g_1, \ldots, g_n,1) \circ (1, \ldots, 1, g_{n+1})) &\\
    & \hspace*{1cm} = \phi_{k_{n+1}+1+\cdots+1} ((f(g_1, \ldots, g_n,1) \tau \circ (g_{n+1},1 \ldots, 1))\hat{\tau}) & \\
    & \hspace*{1cm} = (\phi_{k_1 + \cdots + k_{n}+1}(f(g_1, \ldots, g_n,1) \tau) \circ (\phi_{k_{n+1}}(g_{n+1}),1, \ldots, 1)) \hat{\tau} & \text{by \eqref{eq:claim2}} \\
    & \hspace*{1cm} = \phi_{k_1 + \cdots + k_{n}+1}(f(g_1, \ldots, g_n,1)) \circ (1, \ldots, 1, \phi_{k_{n+1}}(g_{n+1})) & \\
    & \hspace*{1cm} = \phi_{n+1}(f)(\phi_{k_1}(g_1), \ldots, \phi_{k_n}(g_n),1) \circ (1, \ldots, 1, \phi_{k_{n+1}}(g_{n+1}))  & \text{by \eqref{eq:claim2}} \\
    & \hspace*{1cm} = \phi_{n+1}(f)(\phi_{k_1}(g_1), \ldots, \phi_{k_{n+1}}(g_{n+1}))
\end{align*}
as we wanted to show. 
Finally, we show that $\phi$ and $\psi$ preserve the semiclosed structure, i.e., if $s^A$ and $r^A$ denote the section and retraction of the operad $O_A$: 
\begin{equation}
    \label{eq:phi}
    \phi_n \circ s_n = s_n^A \circ \phi_{n+1}\text{,} \quad \phi_{n+1} \circ r_n = r^A_n \circ \phi_n\text{,} \quad s_n \circ \psi_{n+1} = \psi_n \circ s^A_n\text{,} \quad \psi_{n+1} \circ r^A_n = r_n \circ \psi_n \text{.}
\end{equation} 
First of all observe that the second two relations in \eqref{eq:phi} can be obtained from the first two by precomposing with $\psi_n$ and postcomposing with $\psi_{n+1}$. 
The first equation in \eqref{eq:phi} is obvious since $s^A_n$ is the inclusion. 
For the second, observe that by \eqref{eq:comp1} $r_n(f)= \app (f,1)$ for every $f \in C(n)$; therefore,   
    $$\phi_{n+1}(r_n(f))=\phi_{n+1}(\app(f,1)) = \phi_2(\app)(\phi_n(f), \ii) = \mathbf{1}\phi_n(f)\ii = \phi_n(f)\ii=r_n^A(\phi_n(f))$$
as desired. \qed

\section{Axioms defining linear lambda algebras}
\label{app:linear}
The $\lambda^*$-closures of the equations of Table \ref{tab:linear} are the following, in left-to-right, top-to-bottom order:

\begin{itemize}
  \item $\bb \ii \bb = \bb$;
  \item $ \bb(\bb \ii)(\bb \bb \ii) = \bb \bb \ii$;
  \item $ \bb(\bb(\bb \ii))(\cc(\bb \bb(\bb \bb \ii))\ii) = \cc(\bb \bb(\bb \bb \ii))\ii$;
  \item $\bb \ii \cc = \cc$;
  \item $\bb(\bb \ii)(\bb \cc \ii) = \bb \cc \ii$; 
  \item $\bb(\bb(\bb \ii))(\cc(\bb \bb(\bb \cc \ii))\ii) = \cc(\bb \bb(\bb \cc \ii))\ii$; 
  \item $\bb \ii \ii = \ii$; 
  \item $\bb (\bb \ii) \ii = \cc (\bb \bb \ii) \ii$;
  \item $\cc (\bb \cc (\bb (\bb \bb) (\bb (\bb \cc) (\cc (\bb \bb (\bb \cc (\bb (\bb \bb) \ii))) \ii))) ) \ii
        =\cc (\bb \bb (\bb \bb (\bb \cc \ii))) (\cc (\bb \bb \ii) \ii)$;
  \item $\cc(\bb \cc(\bb(\bb \bb)(\bb(\bb \cc)(\cc(\bb \bb(\bb \bb(\bb \bb \ii)))\ii))))\ii 
        = \cc(\bb \bb(\bb \bb(\bb \bb \ii)))(\cc(\bb \bb(\bb \cc \ii))\ii)$;
  \item $\cc (\bb \cc (\bb (\bb \bb) (\bb (\bb \bb) (\cc (\bb \bb (\bb \bb \ii)) \ii))) ) \ii
        =\cc (\bb \bb (\bb \bb (\bb \bb \ii))) (\cc (\bb \bb (\bb \bb \ii)) \ii)$;
  \item $\cc (\bb \cc (\bb (\bb \bb) (\bb (\bb \cc) (\cc (\bb \bb (\bb \cc (\bb (\bb \cc) \ii))) \ii))) ) \ii
        =\cc  (\bb \bb (\bb \cc (\bb (\bb \cc) (\cc (\bb \bb (\bb \cc \ii)) \ii))) ) \ii$;
  \item $\cc (\bb \cc (\bb (\bb \bb) (\bb (\bb \cc) (\cc (\bb \bb (\bb \bb (\bb \cc \ii))) \ii))) ) \ii
        =\cc  (\bb \bb (\bb \cc (\bb (\bb \bb) (\cc (\bb \bb \ii) \ii))) ) \ii$;
  \item $\cc (\bb \cc (\bb (\bb \bb) (\bb (\bb \bb) (\cc (\bb \bb (\bb \cc \ii)) \ii))) ) \ii
        =\cc  (\bb \bb (\bb \cc (\bb (\bb \cc) (\cc (\bb \bb (\bb \bb \ii)) \ii)))) \ii$. 
\end{itemize}

They can be computed by mechanising the operator $\lambda^*$ through a Python script publicly available at \url{https://github.com/arturodefaveri/Linear-lambda-calculus}.

\end{document}